\documentclass[11pt,a4paper]{article}

\usepackage{float}
\usepackage[pagebackref,plainpages=false,colorlinks,linkcolor=blue,anchorcolor=teal,citecolor=blue]{hyperref}
\usepackage{amsmath} 
\usepackage{amsthm} 
\usepackage{amssymb}	
\usepackage{graphicx} 
\usepackage{multicol} 
\usepackage{multirow}
\usepackage[dvipsnames]{xcolor}
\usepackage[margin=1in,bottom=1in]{geometry}
\usepackage[capitalize,noabbrev]{cleveref}

\allowdisplaybreaks[1]

\usepackage[utf8]{inputenc}
\usepackage[english]{babel}

\usepackage{diagbox}
\usepackage{mathtools}

\usepackage{xspace}
\usepackage[T1]{fontenc}
\AtBeginDocument{%
  \DeclareFontShape{T1}{cmr}{m}{scit}{<->ssub*cmr/m/sc}{}%
}
\usepackage{comment}
\usepackage{authblk}

\usepackage{adjustbox}
\usepackage{footnotehyper} 
\makesavenoteenv{table}  

\usepackage{enumitem}
\usepackage{makecell}
\usepackage{threeparttable}
\usepackage{tcolorbox}
\usepackage{tabularx}
\usepackage{booktabs}

\usepackage{thmtools}
\declaretheorem[style=plain,numberwithin=section]{theorem}
\declaretheorem[style=plain,numberlike=theorem]{lemma,corollary}

\declaretheorem[style=plain,numberlike=theorem]{proposition}

\declaretheorem[style=definition]{problem,conjecture}

\usepackage{algorithm}
\usepackage{algpseudocode}

\usepackage{tikz}
\usetikzlibrary{arrows.meta,calc,positioning}

\DeclarePairedDelimiter\rbra{\lparen}{\rparen}
\DeclarePairedDelimiter\sbra{\lbrack}{\rbrack}
\DeclarePairedDelimiter\cbra{\{}{\}}

\DeclarePairedDelimiter{\ket}{\lvert}{\rangle}
\DeclarePairedDelimiter{\bra}{\langle}{\rvert}
\DeclarePairedDelimiter{\abs}{\lvert}{\rvert}
\DeclarePairedDelimiter{\norm}{\lVert}{\rVert}
\newcommand{\ketbra}[2]{\ket*{#1}\!\bra*{#2}}
\newcommand{\braket}[2]{\left< #1 \vphantom{#2} \middle| #2 \vphantom{#1} \right>} 

\newcommand{\binset}{\{0,1\}}

\newcommand{\tr}{\operatorname{tr}}
\newcommand{\rank}{\operatorname{rank}}
\newcommand{\poly}{\operatorname{poly}}
\newcommand{\sign}{\operatorname{sgn}}

\newcommand{\BQP}{\textnormal{\textsf{BQP}}\xspace}
\newcommand{\QSZK}{\textnormal{\textsf{QSZK}}\xspace}

\newcommand{\F}{\mathrm{F}}

\newcommand{\eps}{\varepsilon}

\newcommand{\SV}{\mathrm{(SV)}}
\newcommand{\Uhl}{\mathrm{Uhl}}
\newcommand{\SRAE}{\textsc{SqrtAmpEst}}

\makeatletter
\def\@buildmath#1{%
  \expandafter\def\csname bb#1\endcsname{\ensuremath{\mathbb{#1}}}%
  \expandafter\def\csname bf#1\endcsname{\ensuremath{\mathbf{#1}}}%
  \expandafter\def\csname sf#1\endcsname{\ensuremath{\mathsf{#1}}}%
  \expandafter\def\csname cal#1\endcsname{\ensuremath{\mathcal{#1}}}%
  \expandafter\def\csname rm#1\endcsname{\ensuremath{\mathrm{#1}}}%
  \expandafter\def\csname tt#1\endcsname{\ensuremath{\mathtt{#1}}}%
}
\def\@buildmathletters#1{%
  \ifx#1\relax\else
    \@buildmath{#1}%
    \expandafter\@buildmathletters
  \fi
} 
\@buildmathletters ABCDEFGHIJKLMNOPQRSTUVWXYZabcdefghijklmnopqrstuvwxyz\relax
\makeatother

\begin{document}
\setlength{\abovedisplayskip}{6pt}
\setlength{\belowdisplayskip}{6pt}

\title{Optimal fidelity estimation when one state is pure via algorithmic Uhlmann transform}

\author[1]{Yupan Liu\thanks{Email: \href{mailto:yupan.liu@epfl.ch}{yupan.liu@epfl.ch}}}
\author[2]{Qisheng Wang\thanks{Email: \href{mailto:QishengWang1994@gmail.com}{QishengWang1994@gmail.com}}}
\affil[1]{School of Computer and Communication Sciences, \'Ecole Polytechnique F\'ed\'erale de Lausanne}
\affil[2]{School of Computer Science, Shanghai Jiao Tong University}
\date{}

\maketitle

\begin{abstract}
The Uhlmann fidelity ${\rm F}(\rho_0,\rho_1) = {\rm tr}|\sqrt{\rho_0}\sqrt{\rho_1}|$ is one of the most fundamental quantities in quantum information theory for quantifying the closeness between two quantum states. Estimating the Uhlmann fidelity to within additive error $\varepsilon$ requires a number of copies of the states, or queries to their state-preparation circuits, that depends at least linearly on the smaller of the ranks of $\rho_0$ and $\rho_1$. Consequently, this rank dependence disappears when either state is pure, in which case the query and sample complexities depend only polynomially on $1/\varepsilon$. However, the known optimal estimator for ${\rm F}(\rho,\ket{\psi}\!\bra{\psi})$ due to \hyperlink{cite.FW25}{Fang and Wang (ESA 2025)} requires prior knowledge of which state is pure. 

In this work, we remove this mathematically unnecessary prior-knowledge requirement and establish an optimal estimator for ${\rm F}(\rho,\ket{\psi}\!\bra{\psi})$ under the sole promise that one of the two states is pure, without knowing which one. Our estimator is obtained by specializing the refined algorithmic Uhlmann transform of \hyperlink{cite.UNWT25}{Utsumi, Nakata, Wang, and Takagi (2025)} to the case where one state is pure. In this setting, the Uhlmann fidelity can be recovered as follows: apply a unitary dilation of ${\rm tr}_{\sf A}(\ket{\psi_0}\!\bra{\psi_1})$ (or its inverse) to the reference register $\sf R$ of the purification $\ket{\psi_1}$ (or $\ket{\psi_0}$) on the registers $\sf A$ and $\sf R$, estimate the corresponding square-root amplitude in each case, and take the maximum of the resulting two estimates. 
\end{abstract}

\section{Introduction}
\label{sec:introduction}

The (square-root) Uhlmann fidelity $\F(\rho_0,\rho_1)$, originally introduced by Armin Uhlmann~\cite{Uhlmann76}, is one of the most fundamental measures of closeness between two quantum states:
\[ \F(\rho_0,\rho_1) \coloneqq \tr\abs*{\sqrt{\rho_0}\sqrt{\rho_1}} = \tr\sqrt{\sqrt{\rho_0}\rho_1\sqrt{\rho_0}}. \]
More specifically, Uhlmann fidelity has found applications in various areas within and related to quantum information science, including quantum communication and teleportation~\cite{Schumacher96,BBP+96,HHH99}, the approximate quantum Markov property~\cite{FR15,JRS+18}, information-theoretically secure quantum cryptography~\cite{Mayers97,SP00,TLGR12}, and quantum complexity theory and interactive proofs~\cite{KW00,Watrous02,MY23,BMY25}.

This quantity equals $1$ when the states agree and equals $0$ when their supports are orthogonal. When one state is pure, say $\rho_0=\rho$ and $\rho_1=\ketbra{\varphi}{\varphi}$, the fidelity takes the simplified form
\[ \F(\rho,\ketbra{\varphi}{\varphi}) = \sqrt{\bra{\varphi}\rho\ket{\varphi}} = \sqrt{\tr\rbra*{\rho \ketbra{\varphi}{\varphi}}}. \]
Consequently, when one state is pure, the fidelity can be estimated to constant precision using a constant number of copies of $\rho$ and $\ketbra{\varphi}{\varphi}$ via the SWAP test~\cite{BCWdW01}, yielding an efficient quantum estimator. 
In contrast, for general mixed states, the relation between Uhlmann fidelity and trace distance~\cite{FvdG99} implies that fidelity estimation does not admit an efficient quantum estimator under a standard complexity-theoretic assumption~\cite{Watrous02,Watrous09}.\footnote{The promise problem of estimating the trace distance to within additive error $\eps$ is \QSZK{}-hard~\cite{Watrous02,Watrous09}, and the Fuchs--van Graaf inequalities~\cite{FvdG99} imply that the analogous fidelity-estimation problem is \QSZK{}-hard as well. Therefore, the assumption $\BQP \subsetneq \QSZK$ rules out efficient quantum estimators for both tasks.}

\subsection{Fidelity estimation: the main result and prior work}
In the fidelity estimation task, a quantum estimator is given access to the quantum states $\rho_0$ and $\rho_1$ and outputs an estimate of $\F(\rho_0,\rho_1)$ to within additive error $\eps$. With \emph{purified query access}, an estimator can query state-preparation circuits (including their inverses and controlled versions) for the states, each of which prepares a purification of the corresponding state. With \emph{sample access}, by contrast, an estimator simply receives copies of the states and may measure them collectively.

In this work, we investigate fidelity estimation when one of the states is pure. In the most natural scenario, in which there is no prior knowledge of which state is pure, one-pure-state fidelity estimation can be achieved via the SWAP test~\cite{BCWdW01} with sample complexity $O(1/\varepsilon^4)$ and query complexity  $O(1/\varepsilon^2)$ by combining the SWAP test with quantum amplitude estimation~\cite{BHMT02}. 
Prior to our work, improved query and sample complexities required \emph{additional promises}:
\begin{itemize}
    \item When both states are pure, a query-optimal approach to fidelity estimation was presented in~\cite{Wang24} with query complexity $\Theta(1/\eps)$. 
    \item When it is known \emph{in advance} which state is pure, a query-optimal approach to one-pure-state fidelity estimation was later established in~\cite{FW25}. 
\end{itemize}

The optimal sample complexity $\Theta(1/\eps^2)$ for both tasks was likewise established in those works~\cite{WZ24c,FW26}. 
Here, we remove the requirement for prior knowledge of which state is pure and establish a query-optimal approach to one-pure-state fidelity estimation, quadratically improving the SWAP-test-based query-complexity upper bound in this setting:

\begin{theorem}[Optimal one-pure-state fidelity estimation, informal version of \Cref{thm:query-onePureState-fidelity}]
\label{thm:optimal-one-pure-state-fidelity-informal}
Given purified query access to two quantum states $\rho_0$ and $\rho_1$, under the promise that at least one of them is pure, there is a quantum estimator that, without prior knowledge of which state is pure, estimates $\F(\rho_0,\rho_1)$ to within additive error $\eps$ with query complexity $\Theta(1/\eps)$. 
\end{theorem}

Consequently, combining \Cref{thm:optimal-one-pure-state-fidelity-informal} with the quantum sample-to-query lifting techniques in~\cite[Theorem 1.5]{TWZ25} (see also~\cite[Theorem 1.1]{CWZ25}) yields an estimator for one-pure-state fidelity estimation with sample complexity $\Theta(1/\eps^2)$, also quadratically improving the SWAP-test-based sample-complexity upper bound in this setting.  

\begin{table}[!ht]
\centering
\begin{threeparttable}
\caption{Query and sample complexities of fidelity estimation.}
\label{table:fidelity-quantitative-bounds}
\begin{tabular}{cccccc}
    \toprule 
    & \multirow{2}{*}{Pure states} & \multicolumn{2}{c}{One pure state} & \multicolumn{2}{c}{\multirow{2}{*}{\makecell{Mixed states\\ \footnotesize{(minimum rank $r$)}}}}\\
    \cmidrule{3-4}
    & & \footnotesize{Pure side known} & \footnotesize{Pure side unknown} & \\
    \midrule
    \multirow{2}{*}{\makecell{Query\\ complexity}} & \multirow{2}{*}{\makecell{$\Theta(1/\eps)$\\\footnotesize{\cite{Wang24}}}} & \multirow{2}{*}{\makecell{$\Theta(1/\eps)$\\\footnotesize{\cite{FW25}}}} & \multirow{2}{*}{\makecell{$\Theta(1/\eps)$\\\footnotesize{\Cref{thm:query-onePureState-fidelity}}}} & \makecell{\scriptsize{Upper}\\\scriptsize{Bound}} & \makecell{$\widetilde{O}(r/\eps^2)$\\\footnotesize{\cite{UNWT25}}}\\
    \cmidrule{5-6}
    & & & & \makecell{\scriptsize{Lower}\\\scriptsize{Bound}} & \makecell{$\widetilde{\Omega}(r/\eps)$\\\footnotesize{\cite{Wang2026Nearly}}} \\
    \midrule
    \multirow{2}{*}{\makecell{Sample\\ complexity}} & \multirow{2}{*}{\makecell{$\Theta(1/\eps^2)$\\\footnotesize{\cite{WZ24c}}}} & \multirow{2}{*}{\makecell{$\Theta(1/\eps^2)$\\\footnotesize{\cite{FW26}}}} & \multirow{2}{*}{\makecell{$\Theta(1/\eps^2)$\\\footnotesize{\Cref{cor:samples-onePureState-fidelity}}}} & \makecell{\scriptsize{Upper}\\\scriptsize{Bound}} & \makecell{$\widetilde{O}(r^2/\eps^4)$\tnote{*}\\\footnotesize{\cite{UNWT25}}}\\
    \cmidrule{5-6}
    & & & & \makecell{\scriptsize{Lower}\\\scriptsize{Bound}} & \makecell{$\widetilde{\Omega}(r^2/\eps^2)$\\\footnotesize{\cite{Wang2026Nearly}}} \\
    \bottomrule
\end{tabular}
\begin{tablenotes}
    \footnotesize
    \item[*] The $\eps$-dependence improves further to $O(r^2/\varepsilon^2)$ when either the maximum rank of the states is $r$~\cite{LT26} or the known reference state has rank $r$~\cite{Wang26fidelity}. 
\end{tablenotes}
\end{threeparttable}
\end{table}

To complete our comparison, we also list the prior work for the general case and summarize these results in \Cref{table:fidelity-quantitative-bounds}, where the (minimum-)rank dependence is \emph{unavoidable} in both query and sample complexities. 
In particular, such quantum estimators have been provided recently in~\cite{WZC+23} with query complexity $\poly(r, 1/\eps)$, when the minimum rank of both states is $r$. This upper bound was later improved in a series of works~\cite{WGL+24,GP22,UNWT25}. Query complexity lower bounds for fidelity estimation were investigated in \cite{UNWT25,Wang26fidelity,Wang2026Nearly}. The current best query complexity upper bound is $\widetilde{O}(r/\eps^2)$ due to~\cite{UNWT25} and the query complexity lower bound is $\widetilde{\Omega}(r/\eps)$ due to~\cite{Wang2026Nearly}. In addition to quantum query complexity, the sample complexity for the general case was simultaneously investigated in~\cite{GP22,UNWT25,FW25,LT26,Wang26fidelity,LJ26,Wang2026Nearly}.
The current best sample complexity upper bound with minimum-rank dependence is $\widetilde{O}(r^2/\varepsilon^4)$ due to~\cite{UNWT25}, and the $\eps$-dependence can be further improved to $O(r^2/\varepsilon^2)$ either when the \emph{maximum} rank of both states is $r$~\cite{LT26} or when the known reference state is of rank $r$~\cite{Wang26fidelity}, while the current best sample complexity lower bound is $\widetilde{\Omega}(r^2/\varepsilon^2)$ due to~\cite{Wang2026Nearly}. 

\subsection{Proof technique: Algorithmic Uhlmann transform when one state is pure}

Uhlmann's theorem~\cite{Uhlmann76} (see also~\cite{Jozsa94} for an elementary proof) states that the fidelity $\F(\rho_0,\rho_1)$ equals the maximum overlap between purifications $\ket{\psi_0}^{\sfA\sfR}$ and $\ket{\psi_1}^{\sfA\sfR}$ of $\rho_0$ and $\rho_1$ obtainable by applying a unitary $U$ to the reference register $\sfR$:\footnote{Strictly speaking, $\sign^\SV(X_\Uhl)$ is generally a partial isometry, while any unitary extension attains the same optimal overlap. We omit this subtlety here for simplicity.}
\begin{equation}
    \label{eq:Uhlmann-theorem}
    \F(\rho_0,\rho_1) = \tr\rbra{\abs{\sqrt{\rho_0}\sqrt{\rho_1}}} = \max_U \abs*{ \bra{\psi_0} \rbra[\big]{I^\sfA\otimes U^\sfR} \ket{\psi_1} }.
\end{equation}

An explicit expression for an optimal unitary $U_\star$ that achieves the maximum in \Cref{eq:Uhlmann-theorem}, which we refer to below as the \emph{Uhlmann transform}, appeared implicitly in~\cite[Lemma 6]{Jozsa94}:\footnote{A self-contained proof derived from~\cite{Jozsa94} can be found in~\cite[Appendix F]{UNWT25}; see also Lemma 7.6 in the arXiv version of~\cite{MY23}.}
\begin{equation}
    \label{eq:Uhlmann-transform}
    U_\star = \sign^\SV \rbra*{ X_\Uhl }, \quad\text{where } X_\Uhl \coloneqq \tr_\sfA\rbra[\Big]{ \ketbra{\psi_0}{\psi_1}^{\sfA\sfR} }.
\end{equation}
Here, we refer to $X_\Uhl$ as the \emph{Uhlmann cross operator} and write the simply $X \coloneqq X_\Uhl$ when no confusion can arise, where the sign function is applied to the singular values of $X$. The algorithmic Uhlmann transform, namely the task of approximately implementing the Uhlmann transform $U_\star$, was first investigated in~\cite{MY23} (see also~\cite{BEM+23,BMY25}) using a tailored version of quantum singular value transformation (QSVT)~\cite{GSLW19}. A more efficient and direct approach to implementing $U_\star$, which involves an exact unitary dilation of the Uhlmann cross operator $X$, was later proposed in~\cite{UNWT25} and, when combined with the space-efficient QSVT~\cite{LGLW23}, led to complexity-theoretic consequences in~\cite{LLW25}.

\vspace{1em}
When one state is pure, the Uhlmann transform $U_\star$ admits a much simpler form: since a purification of a pure state is a product state across registers $\sfA$ and $\sfR$, the cross operator $X$ in \Cref{eq:Uhlmann-transform} has rank \emph{at most one}. More specifically, if $\rho_0=\rho$ and $\rho_1=\ketbra{\varphi}{\varphi}$ with a purification of the latter given by $\ket{\psi_1}^{\sfA\sfR} = \ket{\varphi}^\sfA\ket{\eta}^\sfR$, the Uhlmann cross operator can be expressed as
\[ X = \tr_\sfA\rbra*{\ket{\psi_0}\bra{\psi_1}} = \rbra[\big]{\bra{\varphi}\otimes I^\sfR}\ket{\psi_0}\bra{\eta} \coloneqq \ket{\iota_1}\bra{\eta}.\]
Interestingly, the Uhlmann transform $U_\star$ is now proportional to the Uhlmann cross operator:\footnote{Here we assume that $\F(\rho,\ketbra{\varphi}{\varphi})>0$ for simplicity. If $\F(\rho,\ketbra{\varphi}{\varphi})=0$, it holds that $X=0$ and thus every unitary is optimal.}
\begin{equation}
    \label{eq:Uhlmann-transform-onePureState}
    U_\star = \sign^\SV(X) = \frac{\ketbra{\iota_1}{\eta}}{\norm{\ket{\iota_1}}} = \frac{X}{\norm{\ket{\iota_1}}} = \frac{X}{\F(\rho,\ketbra{\varphi}{\varphi})}.
\end{equation}
Here, $\norm{\ket{v}}$ denotes the Euclidean norm of the vector $\ket{v}$. 
To verify that the normalization factor in \Cref{eq:Uhlmann-transform-onePureState} coincides with the Uhlmann fidelity, note that $\rho = \tr_\sfR\rbra*{\ketbra{\psi_0}{\psi_0}}$. A direct calculation then gives
\[\F(\rho,\ketbra{\varphi}{\varphi})^2 = \bra{\varphi} \rho \ket{\varphi} = \bra{\psi_0} \rbra[\big]{\ketbra{\varphi}{\varphi}^\sfA\otimes I^\sfR} \ket{\psi_0} = \norm*{ \ket{\iota_1} }^2.\] 
Therefore, fidelity estimation when one state is pure reduces to estimating the square-root amplitudes associated with the following unnormalized quantum states
\[ \ket{\widehat{\iota}_0} = \rbra[\big]{I^\sfA \otimes X^\dagger} \ket{\psi_0} \quad\text{and}\quad \ket{\widehat{\iota}_1} = \rbra[\big]{I^\sfA \otimes X} \ket{\psi_1}, \]
using the method of~\cite{Wang24}, and then taking the maximum of the two estimates. Here, $\ket{\iota_0}$ corresponds to the case where $\rho_0=\ketbra{\varphi}{\varphi}$ and $\rho_1=\rho$ by similar reasoning. Moreover, when $\rho_j$ is pure for $j\in\binset$, we have $\norm{\ket{\widehat{\iota}_j}^{\sfA\sfR}}=\norm{\ket{\iota_j}^\sfR}$.
Putting everything together, with $\norm{A}$ denoting the operator norm of the matrix $A$, we obtain the alternative expression
\[ \F(\rho_0,\rho_1) = \max\cbra*{ \norm{\ket{\iota_0}}, \norm{\ket{\iota_1}} } = \norm{X}. \]

Consequently, to complete the description of our quantum estimator, it suffices to implement the Uhlmann cross operator $X$ efficiently, as this allows us to efficiently implement the algorithmic Uhlmann transform \emph{without using QSVT} when one state is pure. To this end, the most natural approach is to consider an exact unitary dilation $W$ of $X$ such that $\bra{\bar{0}} W \ket{\bar{0}} = X$. As shown in~\cite[Section 5.1]{UNWT25},\footnote{See also~\cite[Lemma 4.8]{LLW25} for a self-contained statement.} one obtains the following explicit unitary:
\[ W = Q_1^\dagger \rbra[\big]{ I^{\sfA'} \otimes \textup{SWAP}^{\sfR',\sfS} } Q_0. \]
Here, the ancillary registers $\sfA'$ and $\sfR'$ are initialized to and projected back onto $\ket{\bar{0}}$, and the resulting quantum circuit implementation is illustrated by the tensor-network diagram in \Cref{figure:Uhlmann-cross-operator}.

\begin{figure}[!ht]
  \centering
  \begin{tikzpicture}[
      gate/.style={draw,rounded corners=1pt,minimum width=1.05cm,
        minimum height=1.45cm,fill=gray!8,line width=.8pt},
      swaptensor/.style={draw,rounded corners=1pt,minimum width=1.08cm,
        minimum height=1.45cm,fill=ForestGreen!8,line width=.8pt},
      statetensor/.style={draw,rounded corners=1pt,minimum width=1.45cm,
        minimum height=.78cm,line width=.8pt},
      wire/.style={line width=.8pt},
      signal/.style={-{Latex[length=1.8mm,width=1.4mm]},line width=.8pt},
      smalllabel/.style={font=\scriptsize}]

    \begin{scope}[xshift=-7.6cm]
    \node[smalllabel] at (10.85,1.72)
      {\textbf{(a)} Uhlmann cross operator $X$};
    \node[statetensor,fill=BurntOrange!10] (psibar) at (9.45,.25)
      {$\overline{\psi_1}$};
    \node[statetensor,fill=MidnightBlue!8] (psi) at (12.25,.25)
      {$\psi_0$};
    \draw[signal] (8.25,.25) -- (psibar.west)
      node[smalllabel,above,pos=.48] {$\sfS_{\rm in}$};
    \draw[wire] (psibar.east) -- (psi.west)
      node[smalllabel,above,midway] {$\sfA$};
    \draw[signal] (psi.east) -- (13.48,.25)
      node[smalllabel,above,pos=.72] {$\sfS_{\rm out}$};
    \node[smalllabel] at (10.85,-.75)
      {$X=\tr_\sfA\rbra*{\ketbra{\psi_0}{\psi_1}^{\sfA\sfR}}$};
    \end{scope}

    \node at (6.8,.25) {$=$};

    \begin{scope}[xshift=7.2cm]
    \node[smalllabel] at (3.75,1.72)
      {\textbf{(b)} Zero block of the unitary dilation $W$};
    \node[gate,fill=MidnightBlue!7] (qzero) at (1.85,.5) {$Q_0$};
    \node[swaptensor] (sw) at (3.72,-.5) {};
    \node[gate,fill=BurntOrange!8] (qone) at (5.60,.5) {$Q_1^\dagger$};

    \node[smalllabel,anchor=east] (zain) at (.95,1)
      {$\ket{\bar{0}}^{\sfA'}$};
    \node[smalllabel,anchor=east] (zrin) at (.95,0)
      {$\ket{\bar{0}}^{\sfR'}$};
    \node[smalllabel,anchor=west] (zaout) at (6.50,1)
      {$\bra*{\bar{0}}^{\sfA'}$};
    \node[smalllabel,anchor=west] (zrout) at (6.50,0)
      {$\bra*{\bar{0}}^{\sfR'}$};

    \draw[wire] (zain.east) -- ($(qzero.west)+(0,.5)$);
    \draw[wire] (zrin.east) -- ($(qzero.west)+(0,-.5)$);
    \draw[wire] ($(qzero.east)+(0,.5)$)
      -- ($(qone.west)+(0,.5)$);
    \draw[wire] ($(qzero.east)+(0,-.5)$)
      -- ($(sw.west)+(0,.5)$);
    \draw[wire] ($(sw.east)+(0,.5)$)
      -- ($(qone.west)+(0,-.5)$);
    \draw[wire] ($(qone.east)+(0,.5)$) -- (zaout.west);
    \draw[wire] ($(qone.east)+(0,-.5)$) -- (zrout.west);

    \node[smalllabel,anchor=east] (sin) at (.55,-1) {$\sfS_{\rm in}$};
    \node[smalllabel,anchor=west] (sout) at (6.88,-1) {$\sfS_{\rm out}$};
    \draw[signal] (sin.east) -- ($(sw.west)+(0,-.5)$);
    \draw[signal] ($(sw.east)+(0,-.5)$) -- (sout.west);

    \draw[wire] ($(sw.west)+(0,.5)$) -- ($(sw.east)+(0,-.5)$);
    \draw[preaction={draw=white,line width=2.4pt},wire]
      ($(sw.west)+(0,-.5)$) -- ($(sw.east)+(0,.5)$);
    \node[smalllabel,above=1pt of sw]
      {$\textup{SWAP}_{\sfR',\sfS}$};
    \end{scope}

  \end{tikzpicture}
  \caption{Uhlmann cross operator $X$ and the zero block of its unitary dilation $W$.}
  \label{figure:Uhlmann-cross-operator}
\end{figure}
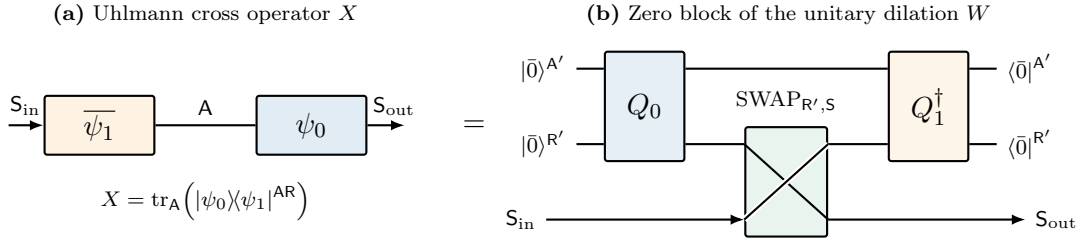

\section{Preliminaries}
\label{sec:preliminaries}

We assume basic knowledge of quantum computation and quantum information theory. For an introduction, we refer the reader to the textbook~\cite{NC10}. 
Throughout the paper, we use the following notation: (1) $\widetilde{O}(f)$ denotes $O(f \operatorname{polylog}(f))$, while $\widetilde{\Omega}(f)$ denotes $\Omega(f/\operatorname{polylog}(f))$; (2) $\ket{\bar{0}}$ denotes $\ket{0}^{\otimes a}$ for an integer $a>1$; (3) we adopt sans-serif letters to denote quantum registers, such as $\sfA$ and $\sfR$. 

\subsection{Square-root amplitude estimation}

We use quantum amplitude estimation~\cite{BHMT02}, in particular the following version:

\begin{lemma}[Square-root amplitude estimation, adapted from~{\cite[Theorem~III.4]{Wang24}}]
\label{lem:srae}
Let $U$ be a unitary, $\Pi$ be a projector, and set $a=\norm*{\Pi U\ket*{\bar{0}}}$. For $0<\eps<1$, there exists an explicit algorithm that outputs $\widetilde a$ such that
\[
  \Pr[\abs*{\widetilde a-a} \leq \eps] \geq 0.99
\]
using $O(1/\eps)$ queries to the controlled versions of $U$ and $U^\dagger$.
\end{lemma}

\subsection{Quantum sample-to-query lifting}

We use the fact that a $q$-query algorithm can be converted into a sample algorithm using $O(q^2)$ samples, which is known as the \emph{quantum sample-to-query lifting}~\cite{WZ25a,WZ25b}. Here, we adopt the following version:

\begin{lemma}[Adapted from~{\cite[Theorem~1.5]{TWZ25}}; see also~{\cite[Theorem~1.1]{CWZ25}}]
\label{lem:query-to-sample}
    Fix a constant number of unknown states. A unitary algorithm using $q$ queries to the state-preparation circuits for each state can, up to any fixed constant error in its output distribution, be simulated using $O(q^2)$ samples of each state.
\end{lemma}

\section{Quantum estimators via algorithmic Uhlmann transform}
\label{sec:algorithm}

\subsection{Algorithmic Uhlmann transform when one state is pure}

Let $Q_0$ and $Q_1$ be the state-preparation circuits for $\rho_0$ and $\rho_1$, respectively, satisfying
\[
    \forall j\in\binset, \qquad Q_j \ket{\bar{0}}^{\sfA\sfR}=\ket*{\psi_j}^{\sfA\sfR} \quad\text{and}\quad \rho_j=\tr_{\sfR}\rbra[\big]{\ketbra{\psi_j}{\psi_j}^{\sfA\sfR}}.
\]
Here, both $Q_0$ and $Q_1$ use a common reference register $\sfR$. Let $\sfS$ be a quantum register that contains the same number of qubits as $\sfR$. Define the unitary dilation $W$ of the Uhlmann cross operator $X$ as
\begin{equation} 
    \label{eq:W-definition}
    W^{\sfA\sfR\sfS} \coloneqq \rbra[\big]{Q_{1}^\dagger}^{\sfA\sfR} \rbra[\big]{I^{\sfA}\otimes \textup{SWAP}^{\sfR,\sfS}} Q_{0}^{\sfA\sfR}.
\end{equation}
In \cite[Section~5.1]{UNWT25} (see also \cite[Lemma 4.8]{LLW25}), it is shown that 
\begin{equation} \label{eq:def-X}
  X \coloneqq (\bra*{\bar{0}}^{\sfA\sfR}\otimes I^{\sfS})W^{\sfA\sfR\sfS}(\ket{\bar{0}}^{\sfA\sfR}\otimes I^{\sfS})
    =\tr_{\sfA}(\ketbra{\psi_0}{\psi_1}^{\sfA\sfR}).
\end{equation}

\begin{lemma}[Uhlmann transform when one state is pure]
\label{lemma:rank-one-Uhlmann-transform}
Let $X=\tr_{\sfA}(\ketbra{\psi_0}{\psi_1}^{\sfA\sfR})$ be the Uhlmann cross operator. Then the Uhlmann transform $U_\star = \sign^\SV(X)$ admits the following simplified form when one state is pure:
\begin{enumerate}[label={\upshape(\arabic*)}]
  \item If $\rho_1=\ketbra{\varphi_1}{\varphi_1}$ is pure, write
  $\ket{\psi_1}=\ket{\varphi_1}^{\sfA}\ket{\eta_1}^{\sfR}$ and let $\ket*{\iota_1}^{\sfR} \coloneqq (\bra*{\varphi_1}^{\sfA}\otimes I^{\sfR})\ket*{\psi_0}^{\sfA\sfR}$.
  Then
  \begin{equation}
    \label{eq:X-rho1-pure}
    X=\ketbra{\iota_1}{\eta_1} \quad\text{and}\quad \norm*{\ket{\iota_1}}=\F(\rho_0,\rho_1).
  \end{equation}
  If $\F(\rho_0,\rho_1) > 0$, then $\ket{\eta_1}$ is the right singular
  vector, $\ket{\iota_1}/\F(\rho_0,\rho_1)$ is the left singular vector, and the
  corresponding singular value is $\F(\rho_0,\rho_1)$.

  \item If $\rho_0=\ketbra{\varphi_0}{\varphi_0}$ is pure, write
  $\ket{\psi_0}=\ket{\varphi_0}^{\sfA}\ket{\eta_0}^{\sfR}$ and let $\ket*{\iota_0}^{\sfR} \coloneqq (\bra*{\varphi_0}^{\sfA}\otimes I^{\sfR})\ket*{\psi_1}^{\sfA\sfR}$.
  Then
  \begin{equation}
    \label{eq:X-rho0-pure}
    X=\ketbra{\eta_0}{\iota_0} \quad\text{and}\quad \norm*{\ket{\iota_0}}=\F(\rho_0,\rho_1).
  \end{equation}
  If $\F(\rho_0,\rho_1) > 0$, then $\ket{\eta_0}$ is the left singular
  vector, $\ket{\iota_0}/\F(\rho_0,\rho_1)$ is the right singular vector, and the
  corresponding singular value is $\F(\rho_0,\rho_1)$.
\end{enumerate}
In either case, we have $\rank(X)\leq 1$ and $\norm*{X}=\F(\rho_0,\rho_1)$.
If $\F(\rho_0,\rho_1) = 0$, then $X=0$.
\end{lemma}

\begin{proof}
We first prove the case where $\rho_1 = \ketbra{\varphi_1}{\varphi_1}$ is pure. Let $\cbra{\ket{a}}_a$ be an orthonormal basis of the quantum register $\sfA$. Then a direct calculation shows that
\begin{align*}    
  X &=\sum_a (\bra*{a}^{\sfA}\otimes I^{\sfR}) \ketbra{\psi_0}{\psi_1}^{\sfA\sfR}
     (\ket*{a}^{\sfA}\otimes I^{\sfR})\\
  &=\sum_a \braket{\varphi_1}{a} \rbra[\big]{ \rbra[\big]{\bra*{a}^{\sfA}\otimes I^{\sfR}} \ket*{\psi_0}^{\sfA\sfR}} \bra*{\eta_1}^{\sfR}\\
  &=\rbra*{ \rbra[\big]{\bra*{\varphi_1}^{\sfA}\otimes I^{\sfR}} \ket*{\psi_0}^{\sfA\sfR} }\bra*{\eta_1}^{\sfR} \\
  &=\ketbra{\iota_1}{\eta_1}.
\end{align*}
Furthermore, another direct calculation gives
\[ 
    \norm*{\ket{\iota_1}}^2 
    =\bra*{\psi_0} (\ketbra{\varphi_1}{\varphi_1}^{\sfA}\otimes I^{\sfR}) \ket*{\psi_0}
    =\bra*{\varphi_1}\rho_0\ket*{\varphi_1}
    =\F(\rho_0,\rho_1)^2. 
\]

We next establish the case where $\rho_0=\ketbra{\varphi_0}{\varphi_0}$ is pure. An analogous calculation with the roles exchanged gives the following equality:
\[ X =\sum_a \braket{a}{\varphi_0} \ket{\eta_0}^{\sfR}\bra{\psi_1}(\ket*{a}^{\sfA} \otimes I^{\sfR}) =\ketbra{\eta_0}{\iota_0}, \]
where the unnormalized quantum state $\ket{\iota_0}$ satisfies that $\norm*{\ket{\iota_0}}^2 =\bra*{\varphi_0}\rho_1\ket*{\varphi_0} =\F(\rho_0,\rho_1)^2$.

In both cases, it is evident that $\rank(X)\leq 1$ and, when $\F(\rho_0,\rho_1)>0$, the unique non-zero singular value of $X$
(i.e., its operator norm $\norm{X}$) equals $\F(\rho_0,\rho_1)$.
Moreover, if $\F(\rho_0,\rho_1)=0$, then $\ket{\iota_1}=0$ in the first case and $\ket{\iota_0}=0$ in the second case, implying that $X=0$.
\end{proof}

\subsection{Query-optimal quantum estimator}

We begin by defining the following quantum circuits $U_1$ and $U_0$, as illustrated in \Cref{figure:U1-circuit,figure:U0-circuit}, respectively: 
\begin{equation}
    \label{eq:U-definitions}
    U_1 \coloneqq
    \rbra[\big]{ W^{\sfA'\sfR'\sfS}\otimes I^{\sfA} }
    \rbra[\big]{ Q_1^{\sfA\sfS}\otimes I^{\sfA'\sfR'} }
    \quad\text{and}\quad
    U_0 \coloneqq
    \rbra[\big]{ (W^{\sfA'\sfR'\sfS})^\dagger\otimes I^{\sfA} }
    \rbra[\big]{ Q_0^{\sfA\sfS}\otimes I^{\sfA'\sfR'} }.
\end{equation}

\begin{figure}[!ht]
\centering
\begin{minipage}[t]{0.49\textwidth}
\centering
\begin{tikzpicture}[
    wire/.style={line width=.75pt},
    prep/.style={draw,rounded corners=1pt,fill=MidnightBlue!7,
      minimum width=1.25cm,minimum height=1.02cm,line width=.75pt},
    oracle/.style={draw,rounded corners=1pt,fill=BurntOrange!8,
      minimum width=1.32cm,minimum height=1.02cm,line width=.75pt},
    block/.style={draw=gray!70,dashed,rounded corners=2pt,line width=.6pt},
    register/.style={font=\scriptsize},
    blocklabel/.style={font=\scriptsize,fill=white,inner sep=1.5pt}]
    \foreach \y in {0,-.62,-1.24,-1.86}
      \draw[wire] (.95,\y) -- (6.85,\y);
    \node[register,anchor=east] at (.82,0) {$\ket{\bar{0}}$};
    \node[register,anchor=east] at (.82,-.62) {$\ket{\bar{0}}$};
    \node[register,anchor=east] at (.82,-1.24) {$\ket{\bar{0}}$};
    \node[register,anchor=east] at (.82,-1.86) {$\ket{\bar{0}}$};
    \node[prep] at (1.55,-.31) {$Q_1$};
    \node[oracle] at (2.95,-1.55) {$Q_0$};
    \draw[wire] (4.30,-.62) -- (4.30,-1.24);
    \draw[wire] (4.20,-.72) -- (4.40,-.52);
    \draw[wire] (4.20,-.52) -- (4.40,-.72);
    \draw[wire] (4.20,-1.34) -- (4.40,-1.14);
    \draw[wire] (4.20,-1.14) -- (4.40,-1.34);
    \node[register,above=1pt] at (4.30,-.52)
      {$\operatorname{SWAP}$};
    \node[oracle,minimum width=1.52cm] at (5.78,-1.55)
      {$Q_1^\dagger$};
    \draw[block] (2.22,-.36) rectangle (6.65,-2.18);
    \node[blocklabel] at (4.44,-2.18) {$W$};
    \node[register,anchor=west] at (6.98,0) {$\sfA$};
    \node[register,anchor=west] at (6.98,-.62) {$\sfS$};
    \node[register,anchor=west] at (6.98,-1.24) {$\sfR'$};
    \node[register,anchor=west] at (6.98,-1.86) {$\sfA'$};
\end{tikzpicture}
\caption{Quantum circuit $U_1$.}
\label{figure:U1-circuit}
\end{minipage}\hfill
\begin{minipage}[t]{0.49\textwidth}
\centering
\begin{tikzpicture}[
    wire/.style={line width=.75pt},
    prep/.style={draw,rounded corners=1pt,fill=MidnightBlue!7,
      minimum width=1.25cm,minimum height=1.02cm,line width=.75pt},
    oracle/.style={draw,rounded corners=1pt,fill=BurntOrange!8,
      minimum width=1.32cm,minimum height=1.02cm,line width=.75pt},
    block/.style={draw=gray!70,dashed,rounded corners=2pt,line width=.6pt},
    register/.style={font=\scriptsize},
    blocklabel/.style={font=\scriptsize,fill=white,inner sep=1.5pt}]
    \foreach \y in {0,-.62,-1.24,-1.86}
      \draw[wire] (.95,\y) -- (6.85,\y);
    \node[register,anchor=east] at (.82,0) {$\ket{\bar{0}}$};
    \node[register,anchor=east] at (.82,-.62) {$\ket{\bar{0}}$};
    \node[register,anchor=east] at (.82,-1.24) {$\ket{\bar{0}}$};
    \node[register,anchor=east] at (.82,-1.86) {$\ket{\bar{0}}$};
    \node[prep] at (1.55,-.31) {$Q_0$};
    \node[oracle] at (2.95,-1.55) {$Q_1$};
    \draw[wire] (4.30,-.62) -- (4.30,-1.24);
    \draw[wire] (4.20,-.72) -- (4.40,-.52);
    \draw[wire] (4.20,-.52) -- (4.40,-.72);
    \draw[wire] (4.20,-1.34) -- (4.40,-1.14);
    \draw[wire] (4.20,-1.14) -- (4.40,-1.34);
    \node[register,above=1pt] at (4.30,-.52)
      {$\operatorname{SWAP}$};
    \node[oracle,minimum width=1.52cm] at (5.78,-1.55)
      {$Q_0^\dagger$};
    \draw[block] (2.22,-.36) rectangle (6.65,-2.18);
    \node[blocklabel] at (4.44,-2.18) {$W^\dagger$};
    \node[register,anchor=west] at (6.98,0) {$\sfA$};
    \node[register,anchor=west] at (6.98,-.62) {$\sfS$};
    \node[register,anchor=west] at (6.98,-1.24) {$\sfR'$};
    \node[register,anchor=west] at (6.98,-1.86) {$\sfA'$};
\end{tikzpicture}
\caption{Quantum circuit $U_0$.}
\label{figure:U0-circuit}
\end{minipage}
\end{figure}

\begin{proposition}[Alternative expression for the Uhlmann fidelity when one state is pure]
\label{prop:no-flag-identity}
Consider the projector $\Pi=I^{\sfA\sfS}\otimes\ketbra{\bar{0}}{\bar{0}}^{\sfA'\sfR'}$. Then one can define the quantities $a_0$ and $a_1$ by
\begin{subequations}
\label{eq:onePureState-fidelity-estimates}
\begin{align}
  a_1 & \coloneqq \norm*{\Pi U_1\ket{\bar{0}}} =\norm*{(I^{\sfA}\otimes X^{\sfS})\ket{\psi_1}^{\sfA\sfS}},\\
  a_0 & \coloneqq \norm*{\Pi U_0\ket{\bar{0}}} =\norm*{(I^{\sfA}\otimes (X^{\sfS})^\dagger)\ket{\psi_0}^{\sfA\sfS}}.
\end{align}
\end{subequations}
When (at least) one of the states $\rho_0$ and $\rho_1$ is pure, we have $\F(\rho_0,\rho_1)=\max\{a_0,a_1\}$. 
\end{proposition}

\begin{proof}
The equalities in \Cref{eq:onePureState-fidelity-estimates} follow immediately from \Cref{eq:U-definitions,eq:def-X}. Using \Cref{lemma:rank-one-Uhlmann-transform}, we obtain the following inequalities:
\begin{subequations}
\label{eq:onePureState-fidelity-estimate-bounds}
\begin{align}
  a_1 &=\norm*{(I^{\sfA}\otimes X^{\sfS})\ket{\psi_1}^{\sfA\sfS}}
    \leq\norm*{X^{\sfS}}
    =\F(\rho_0,\rho_1),\\
  a_0 &=\norm*{(I^{\sfA}\otimes (X^{\sfS})^\dagger)\ket{\psi_0}^{\sfA\sfS}}
    \leq\norm*{(X^{\sfS})^\dagger}
    =\F(\rho_0,\rho_1).
\end{align}
\end{subequations}
It remains to show that at least one of these inequalities is tight:
\begin{itemize}
    \item If $\rho_1 = \ketbra{\varphi_1}{\varphi_1}$ is pure, then $\ket{\psi_1}^{\sfA\sfS}=\ket{\varphi_1}^{\sfA}\ket{\eta_1}^{\sfS}$, and \Cref{eq:X-rho1-pure} gives
    \[  a_1 =\norm[\big]{ (I^{\sfA}\otimes X^{\sfS}) (\ket{\varphi_1}^{\sfA}\ket{\eta_1}^{\sfS})}
        =\norm[\big]{ \ket{\varphi_1}^{\sfA}\otimes \ket{\iota_1}^{\sfS}}
        =\norm[\big]{\ket{\iota_1}^{\sfS}} =\F(\rho_0,\rho_1).
    \]
    \item If $\rho_0 = \ketbra{\varphi_0}{\varphi_0}$ is pure, then $\ket{\psi_0}=\ket{\varphi_0}\ket{\eta_0}$, and \Cref{eq:X-rho0-pure} similarly gives
    \[ a_0 =\norm*{ (I^{\sfA}\otimes (X^{\sfS})^\dagger) (\ket{\varphi_0}^{\sfA}\ket{\eta_0}^{\sfS})}
        =\norm*{ \ket{\varphi_0}^{\sfA}\otimes \ket{\iota_0}^{\sfS}}
        = \norm*{\ket{\iota_0}^{\sfS}} = \F(\rho_0,\rho_1).
    \]
\end{itemize}

Together with \Cref{eq:onePureState-fidelity-estimate-bounds}, we establish the alternative expression $\F(\rho_0,\rho_1)=\max\{a_0,a_1\}$. Note also that the case where $\F(\rho_0,\rho_1) = 0$ is included: in particular, \Cref{lemma:rank-one-Uhlmann-transform} gives $X=0$, so both $a_0$ and $a_1$ are zero.
\end{proof}

We are now ready to present the query-optimal quantum estimator, as given in \Cref{alg:query}:

\begin{algorithm}[!ht]
  \caption{Fidelity estimation when one state is pure.}
  \label{alg:query}
  \begin{algorithmic}[1]
    \Require State-preparation circuits $Q_0$ and $Q_1$ for $\rho_0$ and $\rho_1$, respectively; error parameter $\eps$.
    \Ensure An estimate $\widetilde F$ of $\F(\rho_0,\rho_1)$.
    \State Construct the quantum circuit $W$ from \Cref{eq:W-definition}, as well as the quantum circuits $U_0$ and $U_1$ from \Cref{eq:U-definitions}.
    \For{$j\in\binset$}
      \State Use \SRAE$(U_j,\Pi,\eps)$ to obtain
      $\widetilde a_j$.
    \EndFor
    \State \Return $\widetilde F=\max\{\widetilde a_0,\widetilde a_1\}$.
  \end{algorithmic}
\end{algorithm}

\begin{theorem}[Query-optimal quantum estimator for one-pure-state fidelity estimation]
\label{thm:query-onePureState-fidelity}
Let $\eps(n)$ be an efficiently computable function such that $0 < \eps < 1$. Given purified quantum query access to two unknown $n$-qubit quantum states
$\rho_0$ and $\rho_1$, under the promise that at least one of them is pure, the quantum estimator underlying \Cref{alg:query}, without prior knowledge of which state is pure, outputs an estimate $\widetilde F$ such that
\begin{equation}
    \label{eq:query-guarantee}
    \Pr\sbra*{ \abs*{\widetilde F-\F(\rho_0,\rho_1)} \leq \eps } \geq 0.9,
\end{equation}
using $O(1/\eps)$ queries to $Q_0$ and $Q_1$.
Furthermore, every quantum estimator for this task requires $\Omega(1/\eps)$ queries in the worst case when $0<\eps < 1/4$. 
\end{theorem}

\begin{proof}
Using square-root amplitude estimation (\Cref{lem:srae}), it follows that
\[ \forall j\in\binset, \quad \Pr[\abs*{\widetilde a_j-a_j} \leq \eps] \geq 0.99. \]

By the union bound, we obtain
\begin{align*}
    \Pr\sbra*{ \rbra*{\abs*{\widetilde a_0-a_0} \leq \eps} \wedge \rbra*{\abs*{\widetilde a_1-a_1} \leq \eps} } &\geq 1 - \Pr[\abs*{\widetilde a_0-a_0} > \eps] - \Pr[\abs*{\widetilde a_1-a_1} > \eps]\\
    &\geq 1- 0.01-0.01 = 0.98.
\end{align*}
On this event, by the alternative expression for the Uhlmann fidelity (\Cref{prop:no-flag-identity}), we have
\[ 
    \abs*{\widetilde F-\F(\rho_0,\rho_1)} 
    = \abs*{ \max_{j\in\binset}\widetilde a_j -\max_{j\in\binset}a_j}
    \leq \max_{j\in\binset} \abs*{\widetilde a_j-a_j}
    \leq \eps.
\]
Here, the first inequality uses the elementary inequality $\abs{\max_j x_j-\max_j y_j} \leq \max_j \abs{x_j-y_j}$. We thus establish the bound in \Cref{eq:query-guarantee}.

For the query complexity, it is clear that $U_0$ and $U_1$ use $O(1)$ queries to $Q_0$ and $Q_1$. 
By \Cref{lem:srae}, \Cref{alg:query} uses $O(1/\eps)$ queries to $Q_0$ and $Q_1$ in total, while the matching lower bound follows immediately from~\cite[Theorem V.4]{Wang24}.
\end{proof}

\subsection{Sample-optimal quantum estimator}

Applying \Cref{lem:query-to-sample} with $q=O(1/\eps)$ to \Cref{thm:query-onePureState-fidelity}, we obtain a sample-optimal estimator, while the matching lower bound follows immediately from~\cite[Theorem B.4]{Wang24}. 

\begin{corollary}[Sample-optimal quantum estimator for one-pure-state fidelity estimation]
\label{cor:samples-onePureState-fidelity}
Let $\eps(n)$ be an efficiently computable function such that $0 < \eps < 1$. Given sample access to two unknown $n$-qubit quantum states $\rho_0$ and $\rho_1$, under the promise that at least one of them is pure, there exists an explicit quantum estimator that, without prior knowledge of which state is pure, estimates $\F(\rho_0, \rho_1)$ to within additive error $\eps$ with constant success probability, using $O(1/\eps^2)$ samples of each state.
Furthermore, when $0<\eps<1/4$, every quantum estimator for this task requires
$\Omega(1/\eps^2)$ samples of each state in the worst case.
\end{corollary}

\section*{Acknowledgments}
\noindent The work of Yupan Liu was supported by funding from the Swiss State Secretariat for Education, Research and Innovation (SERI). 
The work of Qisheng Wang was supported by startup funding from Shanghai Jiao Tong University. 
ChatGPT was used interactively to proofread the manuscript, draw circuit diagrams in \Cref{figure:Uhlmann-cross-operator,figure:U1-circuit,figure:U0-circuit}, and identify relevant references, while all writing, including mathematical statements and reasoning, was completed by the authors.

\bibliographystyle{alphaurl}
\bibliography{onePureStateFidelity}

\end{document}